\documentclass[5p]{elsarticle}
\journal{Acta Materialia}

\usepackage[utf8]{inputenc}
\usepackage{color}

\usepackage[colorlinks=true,allcolors=blue,pdfauthor=author]{hyperref}

\usepackage{amsbsy}
\usepackage{amsfonts}
\usepackage{amsmath}
\usepackage{amssymb}
\usepackage{amsthm}
\usepackage{latexsym}
\usepackage{mathtools}
\usepackage{algorithm}
\usepackage{algpseudocode}

\usepackage{graphicx}
\usepackage{epsfig}
\usepackage{epstopdf}
\usepackage{svg}
\usepackage{subfig}

\usepackage{booktabs}
\usepackage{tabularx}
\usepackage{multirow}

\providecommand{\R}{\ensuremath{\mathbb{R}}}
\providecommand{\Z}{\ensuremath{\mathbb{Z}}}

\newtheorem{theorem}{Theorem}[section]

\newtheorem{lemma}[theorem]{Lemma}
\renewcommand{\vec}[1]{\boldsymbol{#1}}
\def\D{{\mathcal D}}

\providecommand{\abs}[1]{\ensuremath{\left\vert#1\right\vert}}

\usepackage{lineno,hyperref}
\modulolinenumbers[5]
\definecolor{darkred}{rgb}{0.8, 0.0, 0.0}
\definecolor{darkgreen}{rgb}{0.0, 0.8, 0.0}
\definecolor{darkblue}{rgb}{0.0, 0.0, 0.8}

\providecommand{\vanish}[1]{}

\graphicspath{ {./figures/} }

\begin{document}

\begin{frontmatter}

\title{Dependence of microstructure classification accuracy on crystallographic data representation}

\author[gm]{Shrunal Pothagoni}
\ead{spothago@gmu.edu}

\author[ucd]{Dylan Miley}

\author[gm]{Tyrus Berry}

\author[ucd]{Jeremy K. Mason\corref{cor}}
\ead{jkmason@ucdavis.edu}

\author[gm]{Benjamin Schweinhart\corref{cor}}
\ead{bschwei@gmu.edu}

\address[ucd]{Department of Materials Science and Engineering, University of California, Davis, Davis, CA 95616, USA}

\address[gm]{Department of Mathematical Sciences, George Mason University, Fairfax, VA 22030, USA}

\cortext[cor]{Corresponding author}

\begin{abstract}
Convolutional neural networks are increasingly being used to analyze and classify material microstructures, motivated by the possibility that they will be able to identify relevant microstructural features more efficiently and impartially than human experts.
While up to now convolutional neural networks have mostly been applied to light optimal microscopy and scanning electron microscope micrographs, application to EBSD micrographs will be increasingly common as rational design generates materials with unknown textures and phase compositions. This raises the question of how crystallographic orientation should be represented in such a convolutional neural network, and whether this choice  has a significant effect on the network's analysis and classification accuracy.
Four representations of orientation information are examined and are used with convolutional neural networks to classify five synthetic microstructures with varying textures and grain geometries.
Of these, a spectral embedding of crystallographic orientations in a space that respects the crystallographic symmetries performs by far the best, even when the network is trained on small volumes of data such as could be accessible by practical experiments.
\end{abstract}


\end{frontmatter}


\section{Introduction}
\label{sec:introduction}

One of the most frequent tasks in materials science is the characterization and analysis of material microstructures.
The reason for this is the understanding that a material's microstructure is a fundamental determinant of the material's properties, and the discipline of materials science is largely concerned with detailizing this understanding.
Given the frequency of this task and the need for expert training to distinguish microstructures that can be visually similar but differ in their properties, there is substantial interest in developing machine learning approaches to efficiently and impartially classify material micrographs.
Of these approaches, convolutional neural networks stand out for their ability to operate on a material micrograph directly without requiring intermediate analysis.

Convolutional neutral networks (CNNs) effectively operate by repeatedly convolving an image with a set of image filters and pooling the results to aggregate spatial information and reduce image aliasing.
This is followed by using a fully-connected neural network on the results of the convolution and pooling layers to perform the final analysis or classification of the image.
The image filters and the parameters of the fully-connected neural network are learned by a training process where annotated images are supplied to the CNN and internal parameters are iteratively adjusted until the CNN returns the expected results.
Depending on the complexity of the images and the nature of the analysis to be performed, the training process can often require a large number of images to adequately constrain the internal parameters.
As a consequence, the comparatively small number of micrographs available in materials science, even using automated acquisition techniques, has generally confined CNNs to being applied to narrowly-defined identification tasks or the classification of micrographs among a small number of broadly-defined categories. 

While a multitude of characterization techniques offer insight into material microstructures, CNNs are generally applied to micrographs of one of three types.
Light optical microscopy (LOM) offers a large field of view and fast data acquisition, but does not directly provide crystallographic orientation or phase information.
Such information can often be inferred from optical contrast though after, e.g., appropriately etching a metallic material.
Scanning electron microscopy (SEM) using secondary electron contrast has a greater depth of focus and can achieve much higher magnifications than LOM, but at the cost of substantially longer inspection times.
While this technique also does not directly provide crystallographic orientation or phase information, the increased electron emission at edges and discontinuities serves to highlight features revealed by etching.
Finally, the scanning electron microscope can also be used to generate electron backscatter diffraction (EBSD) maps of the surface.
Primary electrons that are scattered back toward the detector are elastically diffracted by the crystal lattice along the way.
The resulting diffraction pattern indicates both the phase and crystallographic orientation of the material under the beam, and rastering the beam over the specimen generates an EBSD map where signal discontinuities indicate the presence of a phase or grain boundary.
The resulting phase and grain segmentation is often regarded as the objective ground truth for CNNs trained to perform similar analysis on more easily obtained LOM or SEM images.

As described above, CNNs have mostly been used for segmentation of complex microstructures (usually steels).
An early contribution by DeCost et al.\ \cite{decost2017uhcsdb} upsampled a CNN output to augment greyscale pixel values with information learned by the CNN about the surrounding region.
They used this to successfully segment steels containing grain boundary cementite, ferritic matrix, spheroidite particles, and Widmanstätten cementite with an accuracy of over 90\%.



CNNs
have been shown by Azimi et al.\ \cite{Azimi2018deepclassify} to be able to classify steel optical microstructures with 93.94\% accuracy. 
CNNs can be expanded to Convolution Deep Belief Networks (CDBN) using unsupervised learning on an unlabeled training set in order to find common features.
Cang et al.\ \cite{cang2017deepbeliefCNN} applied a CDBN to extract features from multiple material systems and then used the features to reconstruct the microstructure.
However, the classification capabilities of CNNs and CDBNs require careful network parameter selection and sufficient training data.
Another example by Tsutsui et al.\ \cite{tsutsui2020semcompare} was applied steel SEM microstructures with different processing paths with 85\% accuracy using a feature extraction and decision tree based machine learning method.
However, the careful selection and preparation of the training set limits these machine learning driven techniques to specific systems of microstructures and may cause them to output nonsensical results if the training sets are insufficient or biased.
The work by DeCost et al.\ \cite{decost2017uhcsdb} shows CNNs trained on features extracted from SEM images via both an intensity gradient detection technique (SIFT) and by transfer learning, which uses convolutional filters taken from neural networks that were trained on photographs (dogs, cats, etc.) rather than on materials data.
While these methods yielded classification accuracy in the low to upper 90s, the performance of their CNNs required large amounts of data and were unable to reliably form processing-structure correlations, possibly in part due to the arbitrarity of the descriptors used.
In all, these works have pushed the envelope for materials data science.
Despite the successes of these foundational works, the volume of data required will hinder widespread feasibility of these methods.

For the purposes of this paper, microstructure is defined as the information that is available at the level of individual grains or phase domains.
Since this paper is concerned specifically with the representation of crystallographic information, the microstructures considered below are further restricted to single phase face-centered cubic (FCC) materials, though extension of our results to other crystal symmetries or to multi-phase materials is straightforward.

\subsection{Microstructure Classification}

The lens through which materials have been viewed historically --- and to an extent presently --- is founded upon the fabrication pathways that have been followed. 
The classification and standardization of materials by the means through which they were created naturally results in the use of this lens and vice versa.
As our understanding of the intricacies of the microstructure have developed, we have been able to begin to engineer trajectories through the microstructure state space in order to more intentionally fabricate materials with specified structure and properties.
However, materials design and discovery is still largely confined to small modifications of previously explored pathways. 
Great efforts have been made to elucidate the myriad mysteries of the microstructure on all fronts of the field, ranging from instrumentation and fabrication to theory and mathematics.
The result of such efforts has resulted in the ability to thoroughly control, test, and characterize materials at the microstructure level and deeper. 
However, the true correlations between fabrication techniques, structure, and physical properties remain poorly understood.
To interface the aforementioned advancements and create a holistic, integrated understanding of materials it is necessary to develop generalizable materials representations and classification methodologies. 

Previously, the idea of a microstructure state space was formally introduced as a means to describe the generalized interconnection of materials and all associated quantities and qualities \cite{miley2024}.
The bridge between this concept and reality is the construction of a materials database which represents materials via truly universally characteristic aspects in a minimally reduced form.

Given that the nature and extent of the information required to fully describe the state of a material is unclear, it is necessary at the outset to clearly define the information of interest and the corresponding representation.
Arriving at such a minimally-reduced representation paves the way for inter-class compatibility and interoperability of any dataset which contains the materials represented with the same characteristic dimensions e.g. the comparison of grain structure geometry regardless of differences in chemical species. 
This representation should be minimially reduced in the sense that it retains sufficient information to usefully constrain the material properties, and perhaps even to reconstruct the material in a statistical sense.
Generalized microstructure metrics can be used to compare the microstructures in the database by giving insight into the distance or similarity between them.
Concurrently, such a database would enable holistic material categorization based upon all relevant information available for microstructures, furthering the epistemological depth of the field as a whole.
Correct means of classifying microstructures would enable the correlation of properties with specific microstructures and even with the constituent aspects of given microstructures. 
The correctness of the classification does not solely rely on the method used, but is also deeply dependent on the information that is operated on. 
Just as information shapes human conceptualization, so too does it shape machine classification and thus the ability to make useful property correlations.

An open ended question remains: is it possible to provide a unique specification for a material that is simple enough to be practicable yet retains enough information to distinguish materials classes, predict material properties, and inform processing techniques?
Neural networks may enable the answer to such a question to be found.
Convolutional neural networks (CNN) in particular are a promising tool due to their ability to extract reduced dimension features that can be used to characteristically represent images. 
CNNs have been implemented widely in science and engineering for a vast range of applications, including within materials science. 

Before considering how to remedy the small-data issues inherent to much of materials science, it may be worthwhile to consider the ways that materials are represented. 
At the level of microstructure, micrographs encode some extent of feature data into each pixel/voxel such as a grayscale integer corresponding to the intensity of detected information and correlated crystal plane orientation information given by an accompanying electron dispersive x-ray spectroscopy (EDS) scan.
These two forms of data allow offer insight into the spatial arrangement and morphology of grains with their associated crystal orientations.
Often implicitly, it is assumed that the single phase microstructures are polycrystalline, given that the single-crystal single-phase microstructure is a trivial limiting case to describe and classify.
The breadth of microstructures considered can be expanded to include multi-phase materials, though these cases are not treated in this work.

The possible crystallographic planes that arise in the microstructure correspond to the local rotational symmetry made available by the atomic ordering (FCC, BCC, HCP etc.).
The specific reference lattice for a given phase will have different extents of degeneracy for rotations of the lattice.
There can be different extents of degeneracy for materials with different reference lattices, e.g. cubic vs triclinic systems.

\subsection{Crystallographic Orientation}
\label{subsec:crystal_orientation}


Given a reference lattice, an element of the rotation group $\text{SO}(3)$ corresponds to a grain orientation. Specifically, $M \in SO(3)$ corresponds to the crystalline lattice orientation if rotating a crystal initially in the reference orientation by $M$ brings the crystal to the observed orientation. However, this representation contains redundancies: two elements of $\text{SO}(3)$ correspond to the same grain texture if they differ by the action of the point symmetry group $H$ of the crystalline lattice. Instead, crystalline orientations are in bijection with a the space of equivalence classes rotation matrices under the action of $H$. Formally, the \emph{quotient space} $\text{SO}(3)/H$ consists of equivalence classes $[Q]=\{hQ\colon h\in H\}$. $\text{SO}(3)/H$ is an example of a homogeneous space, which are the subject of an expansive literature in geometry.

If we were to represent a micrograph by an array of elements of $SO(3)$ $(Q_{i,j})_{i,j=1\ldots,N}$ (e.g. by a vector of unit quaternions), the array would exhibit both global and local symmetry. The global symmetry action of $SO(3)$ sends $(Q_{i,j})_{i,j=1\ldots,N}$ to $(P Q_{i,j})_{i,j=1\ldots,N}$ for an element $P\in SO(3)$ (that is, it acts by changing the reference direction). The local symmetry action of $H$ on $(Q_{i,j})_{i,j=1\ldots,N}$ acts individually on each grain (or even on each voxel) to permute redundant representations of the same crystallographic texture.   

While all of the examples considered here contain crystals with cubic point group symmetry, the approach naturally generalizes to arbitrary centrosymmetric crystals.
It is also applicable to data sets where the orientation is not constant on grains or ones exhibiting multiple phases. 


\subsection{Grain Data Vector Representations}
\label{subsec:GDV_reps}

We append phase information to micrographs by associating each voxel with a vector indicating the orientation of the grain that it belongs to via a function $f:SO(3)/H\to \mathbb{R}^n$ for some $n$.
We call this a \textbf{Grain Vector Data} (GVD) representation.
Note that this description assumes that grain boundaries have infinitesimal width and are defined implicitly, and that orientations are constant on grains.
We describe four different GVDs, two of which are widely used in the materials science literature.

We begin by introducing a few desired properties of a GVD.
Ideally, a GVD will provide a unique representation of crystalline orientation which depends continuously on the orientation, and so that different GVDs correspond to distinct orientations.
In mathematical language, an ideal GVD would define a continuous, injective (one-to-one) function $f: SO(3)/H$ to $\R^n.$
Observe that $h:SO(3)\to \R^n$ induces a well-defined function on $SO(3)/H$ precisely when it is $H$-invariant: $f(h^{-1} g)=f(g)$ for all $g\in SO(3)$ and all $h\in H.$
$H$-invariance is not to be confused with $SO(3)$-equivariance: a function $f:SO(3)/H\to \R^n$ has this property if there is an action of $SO(3)$ on $\mathrm{im} f$ so that rotating the input to $f$ is equivalent to applying the same rotation to the output.
For example, the function $k:\R^2\setminus \left\{0\right\}\to S^2$ assigning to a vector a unit vector in the same direction is $SO(3)$-equivariant.
Equivariance is a desirable property in machine learning, as it results in models which respect the underlying symmetry of the system; see Section~\ref{subsec:CNN_properties} for more discussion.

The Feature ID (FID) is a GVD which assigns distinct integers to each grain.
Use of this GVD allows for the implicit description of morphology and disregards crystallographic information.
It does not induce a well-defined function on grain orientation space. 

The most popular GVD in the materials science literature is the Inverse Pole Figure (IPF).
Given a specimen crystal and a reference lattice, a unit normal vector $\vec{n}$ of the specimen generically corresponds to $\abs{H}$ normal vectors in the reference, where $H$ is the point symmetry group of the lattice.
In the case of cubic symmetry, only one of these vectors is located in the geodesic triangle on $S^2$ formed by the vectors $\vec{e}_1,$ $\vec{e}_1+\vec{e}_2,$ and $\vec{e}_1+\vec{e}_2+\vec{e}_3.$
The IPF associated to the specimen crystal is then the stereographic projection of the unique vector in this triangle corresponding to $\vec{n}.$
It accounts for the redundancy created by the point symmetry group in the special case of cubic symmetry (and thus induces a well-defined, continuous function on $SO(3)/H$), but loses information by projecting a three-dimensional space to a two-dimensional one.
It is $SO(3)$-equivariant.
This is discussed further in Section \ref{sec:discussion}.

A third GVD is the quaternion parametrization (QP). The unit quaternions are in bijection with the special unitary group $SU(2)$ via the transformation
$$a\vec{1}+b\vec{i}+c\vec{j}+d\vec{k}\leftrightarrow \begin{bmatrix} a+bi & c+di \\ -c+di & a-bi \end{bmatrix}\,.$$
Identifying $\R^3$ with $\mathrm{Span}\{\vec{i},\vec{j},\vec{k}\},$ a unit quaternion $q$ acts on $v\in\R^3$  by 
$$v\mapsto qvq^{-1}\,,$$
which, by direct computation, is a rotation action. It is easily checked that this induces a $2$-to-$1$ mapping from the unit quaternions to $SO(3),$ with $q$ and $-q$ being mapped to the same element. While it is uncommonly used in the materials science community, the quaternion representation of $SO(3)$ sees frequent use in computer science, aerospace engineering, and physics due to its condensed form, completeness, and computational practicability. Note that --- without further work --- the quaternion representation does not account for the degeneracy corresponding to the action of the point symmetry group and thus does yield a well-defined function on grain orientation space. It is, however, injective in the sense that distinct quaternions correspond to distinct orientations.

The spectral embedding (SEQ) gives a unique representation for each equivalence class of crystallographic orientations under the point symmetry group via an orthogonal basis for the function space $L^2\left(SO\left(3\right)/H\right)$  $\left\{f_1,f_2,\ldots\right\}.$ Given such a basis we obtain an $n$-dimensional $GVD$ by assigning to a crystalline orientation $g\in SO(3)/H$ the vector $(f_1(g),\ldots,f_n(g)).$ Here, the functions $f_i$ are computed in terms of the coefficients of the Wigner $D$-matrices. See the appendix for an explanation of how to find such a basis in a more general setting. Equivalently, a basis for $L^2\left(SO\left(3\right)/H\right)$ may be found by identifying the unit quaternions with $S^3$ and manipulating the hyperspherical harmonics on that space \cite{MASON20086141}. 

For cubic symmetry, a direct computation shows that the first nine non-constant functions we obtain from the spectral embedding gives an immersion of $SO(3)/H$ into $\mathbb{R}^9$ (that is, a differentantable function whose Jacobian is non-degenerate). However, in practice, we use the first four non-constant basis functions for the spectral embedding.

See Table~\ref{table:GVDs} for a summary of the properties of the different representations.

\begin{table*}
\centering
\begin{tabular}{|l|c|c|c|c|c|}
\hline
Representation & FID & IPF & QP & SEQ (Full) & SEQ (truncated) \\ \hline
Dimensionality (n)              & 1   & 2   & 4  & 9          & 4               \\ \hline
Single-valued  &     & X   &    & X          & X               \\ \hline
Generically injective      &     &     & X  & X          &                 \\ \hline
Continuous      &     &  X    &   & X          &  X                \\ \hline
$SO(3)$-Equivariant      &     &  X    &X   & X          &  X                \\ \hline
\end{tabular}
\caption{\label{table:GVDs} A summary of the properties of the GVDs, where a GVD is viewed as a (multi)function $f:SO(3)/H\to \mathbb{R}^n$ and $H$ is the cubic symmetry group. We say that such a function is single-valued if it assigns a single element of $\R^n$ to each element of $SO(3)/H.$ Note that the full spectral embedding is an immersion, which implies that it is injective except on a subset of $SO(3)/H$ of measure zero.}
\end{table*}



\subsection{Convolutional Neural Networks}
\label{subsec:CNNs}

A Convolutional Neural Network (CNN) is a standard tool from the field of computer vision often used for image classification tasks \cite{aghdam2017guide}. 
The basic CNN architecture consists of two main components: the feature extraction blocks and a multilayer perceptron (MLP).


As the terminology suggests, the \textit{convolution operator} is fundamental to  CNN architecture, and in particular to the feature extraction process.
It is a discrete analogue of the classical notion with the same name: given a signal $f(x)$ and kernel $k(x)$, a convolution operator is defined by \[ [f\star k](t) = \int_{-\infty}^{\infty} f(t-\tau)k(\tau)d\tau \]
For a fixed $t,$ observe that $[f\star k](t)$ is the inner product of $f(t-\tau)$ and $k(\tau)$ in the space of square integrable functions.
As a result, the convolution  operator may be interpreted as a continuous measure of the local similarities of $f(t)$ with $k(t)$ as we translate the kernel over the signal.
Analogously, the \textit{discrete convolution operator} is defined on a 2D array (an image composed of pixels) $f(x,y)$ with filter $k(x,y)$ by \[\label{eqn:conv}[f \star k](x,y)=\sum_{dx \in \Z}{\sum_{dy \in \Z}{ f(x-dx,y-dy)k(dx,dy)}} \tag{1}\]
Intuitively, this discrete operator acts on the 2D array by taking the inner product of $k(x,y)$ with $f(x,y)$ as $k(x,y)$ slides across $f(x,y)$.
\footnote{As $f(x,y)$ is a scalar valued function, the provided convolution operator definition is only compatible for grey-scaled images.
This definition can be further generalized for images with multiple color maps such as RGB or CMYB by instead representing $f(x,y)$ as a vector valued function, however, this will alter the equation slightly as we would need to use distinct kernels over each of the corresponding color maps.}
It is important to note that the parameters of $k(x,y)$ are fixed, known as \textit{weight sharing} \cite{pmlr-v48-cohenc16}, and it enables the convolution operation to detect features in an image regardless of locality.
the resulting image after convolution is known as a \textit{feature map}.
The resulting feature map is then subsampled (pooled) to reduce the reliance of exact positioning of features.
This process of convolving and pooling an image as illustrated in Figure \ref{fig:conv_pool} forms an individual feature extraction block.
Most CNN architectures use multiple blocks to learn a multitude of distinct higher order features that may be present in an image \cite{726791}. 

Once the image has passed through a series of feature blocks, it is flattened and propagated forward through the MLP.
The MLP is explicitly a composition of affine and nonlinear transformations of the from $\sigma(\textbf{W}_i\textbf{x}+\textbf{b}_i)$ where \[\textbf{W}_i \colon \R^m \to \R^n\] and $\textbf{b}_i \in \R^n$.
Geometrically, as the data propagates through the MLP, the composition of affine and nonlinear transformations acts to stretch and pulling apart the data into distinct clusters which are then separable by hyperplanes. 

The CNN's ability to accurately classify data is predicated on the right choice of $\textbf{W}_i$, $\textbf{b}_i$, and $k(x,y)$.
However, it is impossible to determine how to initialize these parameters.
Instead, the parameters are randomly initialized are then fined tuned using a method known as \textit{backpropagation}.
A comprehensive analysis of backpropagation is available in \cite{schmidhuber2015deep}, but in short, backpropagation is a gradient based method which updates the model parameters by minimizing the error of the loss function.

\begin{figure}
    \centering
    \includegraphics[width=\columnwidth]{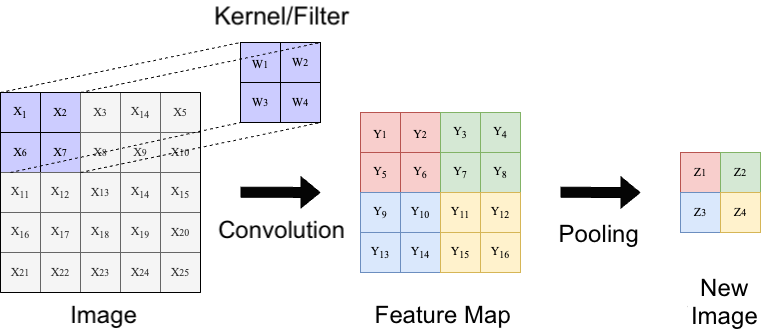}
    \caption{\label{fig:conv_pool} General pipeline for how an image is processed in the feature extraction block of a CNN.}
\end{figure}

\subsection{Invariant and Equivariant Machine Learning}
\label{subsec:CNN_properties}

The benefit of using a CNN is that the convolution operator defined in (\ref{eqn:conv}) is \textit{translationally equivariant} in the sense that translations applied to the input are mapped to the output.
To make this precise, given a pixel array $f(x,y)$ and a translation $g\in \Z^2,$  denote by $[T_gf](x,y)$ the image whose pixels are translated by $g.$
Then \[[[T_g f]\star k](x,y) = [T_g[f\star k]](x,y)\,.\]
In words, the convolution of the translate is the translate of the convolution, a consequence of the  aforementioned weight sharing property.
It is also desired that translations of objects applied to our image do not effect the classification of the image. More explicitly, we say that a classifier $\Phi$ on the space of 2D image arrays is \textit{translationally invariant} if $\Phi(f(x,y))$ satisfies, \[\Phi(T_g  f(x,y)) = \Phi(f(x,y))\]
for all images $f(x,y).$ 

These properties have led to the development of a subfield of machine learning known as equivariant machine learning.
It aims to generalizes the CNN construction to create machine learning models that are equivariant with respect to general group actions.
However, the specific architecture proposed in \cite{pmlr-v48-cohenc16, cohen2019general} is not directly applicable to micrographs because they are not themselves elements of an $SO(3)$-homogeneous space but rather arrays of elements of the $SO(3)$-homogeneous space $SO(3)/H$ ($SO(3)$ acts on this array, but not transitively).
Instead, we generalize the spectral approach of \cite{esteves2018learning} to define a Grain Vector Data representation that is both $H$-invariant (redundancy free) and $SO(3)$-equivariant.
By training a traditional CNN on micrographs represented using this GVD, we obtain models that are equivariant to both the $SO(3)$-action on the reference lattice and the $\Z^2$-action of translations.

\subsection{Summary of Results}
\label{subsec:results_summary}

The aim of our study is to compare the performance of four different GVD representations: FID, IPF color, phase-quaternion, and spectral embedding.
For each such representation, we optimized multiple CNN architectures for a classification task on a data set consisting of five synthetic materials of differing grain geometry and texture.
Our results demonstrate that the spectral embedding yields the best results with respect to classification accuracy, even for limited sample sizes.

\section{Synthetic Microstructures}
\label{sec:synthetic}

The data sets studied in this paper are synthetic microstructures generated by the software package  DREAM.3D \cite{groeber2014dream}.  It samples random 3D microstructures matching prescribed morphological and crystallographic parameters. \,  
Specifically, a 3D grid of voxels is segmented into a grain network in which each voxel is labeled with crystal orientation and phase information.
The morphology of the grain network is largely determined by an equivalent sphere diameter (ESD) distribution and an axis-orientation distribution function (ODF), which govern the average grain size and growth anisotropy, respectively.
Crystal orientation is assigned based upon an ODF over the stereographic projection of crystal planes, described and plotted graphically on pole figures with respect to 3 zone axes: {$\langle$}001{$\rangle$}, {$\langle$}011{$\rangle$}, and {$\langle$}111{$\rangle$}. 
The orientation of neighboring grains are controlled via a misorientation distribution function.
Combined, these parameters allow for the generation of varying degrees of crystallographic texture.

\begin{figure}
    \centering
    \subfloat[Equiaxed \label{equiaxed}]{\includegraphics[width=2.8cm]{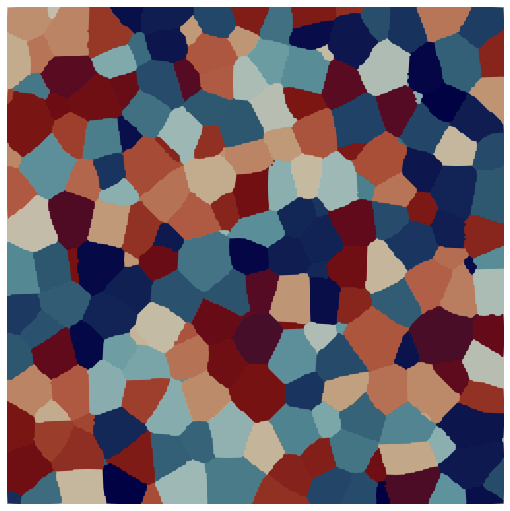}
    }
     \subfloat[Rolled \label{rolled}]{\includegraphics[width=2.8cm]{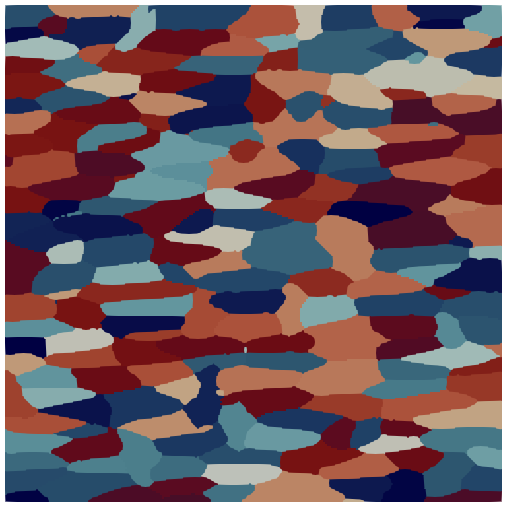}}
     \subfloat[Bimodal \label{bimodal}]{\includegraphics[width=2.8cm]{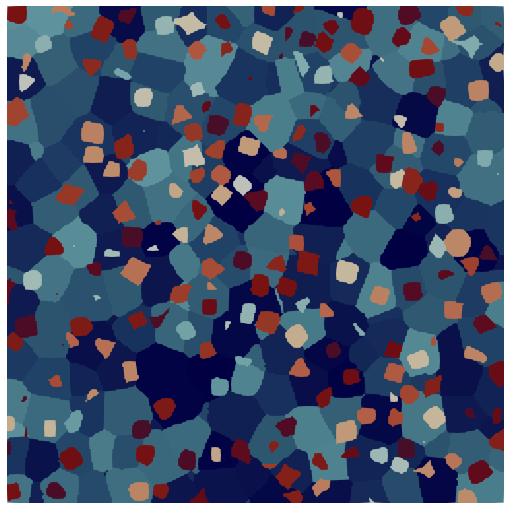}}\\
     \subfloat[Untextured \label{untextured}]{\includegraphics[width=5.6cm]{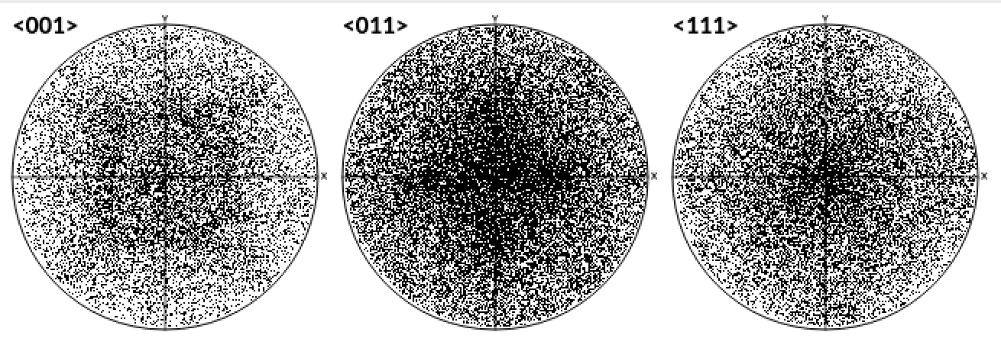}}\\
     \subfloat[EI-Cu Texture \label{cu}]{\includegraphics[width=5.6cm]{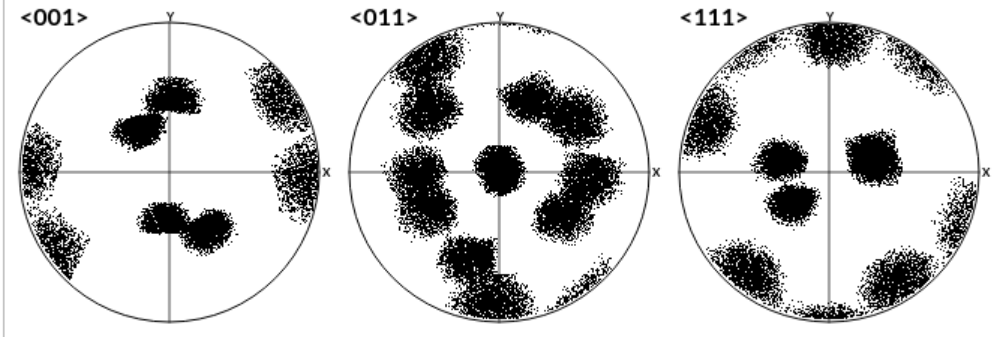}}\\
     \subfloat[EI-BG Texture \label{bg}]{\includegraphics[width=5.6cm]{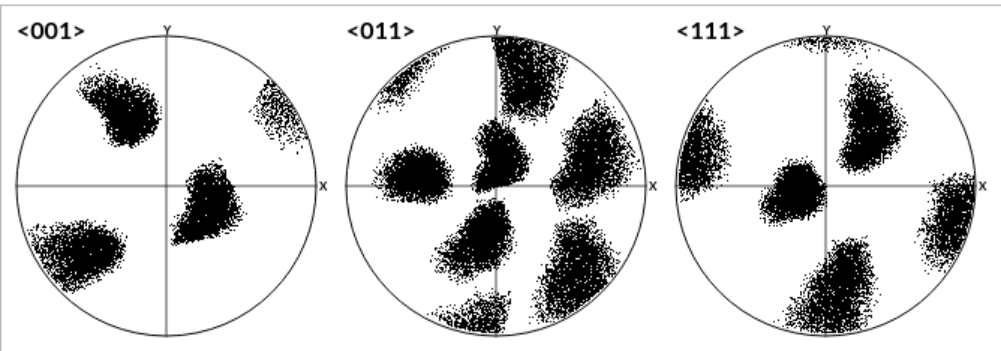}}
    \caption{Example grain morphologies and textures from the synthetic microstructure data set.}
    \label{fig:enter-label}
\end{figure}

A set of five structures were generated with DREAM.3D to explore the performance of the various GVDs for classification tasks.
Three of the five materials are untextured and have morphologies corresponding to i) equiaxed (EI, Fig.~\ref{equiaxed}) ii.) rolled (RI, Fig.~\ref{rolled}) and iii.) bimodal (BI, Fig.~\ref{bimodal}) structures.
The other two have morphologies corresponding to EI structures; we call them EI-Cu and EI-BG. 
The ESD is the same for each material, except the secondary phase of BI which has roughly half of the ESD of the matrix and occupies 25\% of the total volume. 
RI is heavily biased to grow along specific morphological axes.
Different crystallographic texture was introduced for the final two materials, EI-Cu and EI-BG. They are depicted in Figures~\ref{cu} and~\ref{bg}, respectively. 
The grain orientations for EI-Cu arise from an ODF which combines a set of three textures; “copper”, “brass”, and “S”. These textures can be described by the Euler angle sets (90°, 35°, 45°), (35°, 45°, 0°) and (59°, 37°, 63°). 
The EI-BG texture is based on a combination of the “brass” and “Goss” textures, which have Euler angles of (35°, 45°, 90°) and (0°, 45°, 90°), respectively. 
These materials are intended to force the latent microstructure representation learned by the CNN to include grain orientation information, rather than allowing the CNN to rely entirely on grain morphology information.

\begin{figure}
    \centering
    \includegraphics[width=\columnwidth]{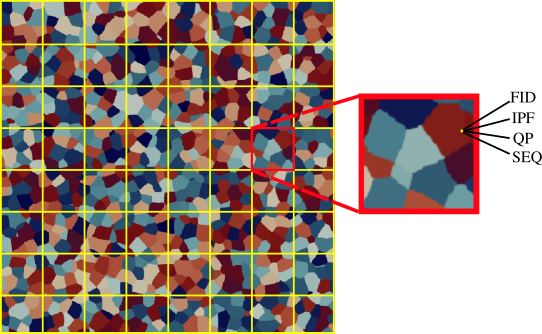}
    \caption{$512 \times 512$ slices of a synthetic material are sampled via 8 level-set intersections, creating $64$ $64\times 64$ square windows from each slice.
    Slices of the sample are taken every 10th voxel along the z-direction, which is a rough approximation of the average ESD of the matrix.
    Each voxel within each window has 4 associated GVDs (FID, IPF, QP, and SEQ).}
    \label{fig:subimg}
\end{figure}

We assess the suitability of the four GVDs by training CNNs on micgrographs augmented with each representation. We train and optimize the models separately for each GVD, and on four separate batches of data that are increasing in sample size.
The batches of data are generated as illustrated in Figure \ref{fig:subimg} using increasing numbers of level-set slices.
We expect that the CNN will identify and learning the underlying geometry and texture of the dataset at different rates depending on the GVD representation.
Determining this rate of convergence among the GVD is necessary as physical micrograph data sets obtained by EBSD are generally limited in size due to the milling techniques that are used to remove material layers.
Hence, the best GVD based representations is one allows the CNN to most accurately distinguish between material classes with few training samples. 

The limitations imposed on classification mechanisms by material representation are investigated via training the CNN architecture on the 4 different GVDs described in section \ref{subsec:GDV_reps}. GVDs which only consider grain structure (that is, the FID  labels) will be unable to distinguish the three structures which differ only in crystalline texture, despite the immense  possible differences in properties.

CNNs trained on data sets using any of the GVDs will be able to detect variations in geometry. 
Even though there is no explicit label stating that voxels belong to one grain or another, grain boundaries are represented by adjacent voxels with differing orientation descriptor values.

\section{CNN Methodology and Results}
\label{sec:cnn}

\subsection{Data processing}

The five classes of synthetically generated materials were retrieved from DREAM.3D as text files with their associated GVD representations, except the spectral embedding. That GVD was computed from the quaternion coordinates. 
\footnote{All of the code used to generate and preprocess the data is available for public use through our GitHub: \hyperlink{code}{https://github.com/Shrunalp/cnn-mat.git}.}

We represent the material as a four-dimensional array, where the first three indices correspond to the $x,y,$ and $z$ position and the last four give the GVD at that point in space.
The coordinates used to describe each voxel is of an arbitrary physical size, however, they do represent the typical resolution of an EBSD scans' range, between 20-50 nm \cite{GINTALAS2022648}.
We call this array a \textit{block form} representation of the material.
In order to compare the different GVDs using the same network architecture, we modify the input data so that the dimensionality of each representation is the same.
To accomplish this, the FID integer is duplicated four times and the IPF coloring is duplicating twice, resulting in a four-dimensional vector.



For each material class, we independently generate a synthetic microstructures of dimension $512\times 512 \times N \times 16$ with \[N \in \{128, 256, 512, 1024\}.\] Upon creation, the array is then segmented into four distinct $512\times 512\times N \times 4$ arrays for each material class, each with a distinct GVD representation. 

Each block form array is transformed into 2D image data by selecting $512 \times 512 \times 4$ slices (where the last four coordinates correspond to the GVD data) at increments of $10$ along the $z$-axis.
The spacing is chosen to minimize statistical dependence between cross-sections as the average grain diameter is 10 pixels.
These images are then segmented into 64 separate \, $64\times 64\times 4$ sub-images as demonstrated in Figure \ref{fig:subimg}. 
All the sub-images are split into a $64/16/20$ training, validation, and testing set. 


\subsection{Architecture design}
\label{subsec:architecture_design}

Given that there are four distinct GVD representation datasets, we optimize and adapt our CNN model architecture to each specific representation for a fair comparison.
This is accomplished by performing a hyperparameter sweep on the CNN architectures of each block dataset and their associated GVD representation.

\begin{figure}
    \centering
    \includegraphics[width=\columnwidth]{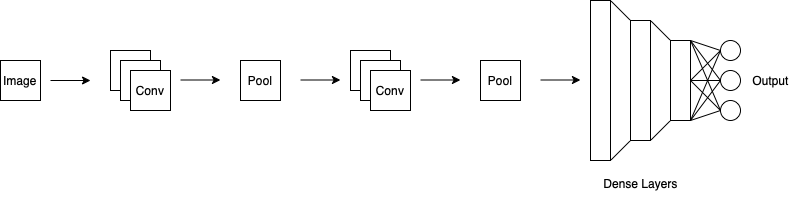}
    \caption{\label{fig:cnnarch} General CNN model framework used for classification.}
\end{figure}

The CNN models used in these experiments are created using the popular machine learning library TensorFlow \cite{tensorflow2015-whitepaper}.
We require that all models must include two separate convolution and pooling layers and three dense fully connected layers in the MLP as illustrated in Figure \ref{fig:cnnarch}.
We also fix the size of all of the convolution and pooling layers with filters to be $3 \times 3$ and use the ReLu activation function in all but the last layer.
The last layer then utilizes the softmax activation function for classification.
We employ the popular stochastic optimization algorithm Adam \cite{kingma2014adam} for training the CNN models with a fixed learning rate of $\alpha = 0.001$.
The  loss function for our multiclass image dataset is categorical cross entropy which is given by the formula \[CE = - \sum_{i=1}^{N}y_i\log(\hat{y_i})\] where $y_i$ is the true value and $\hat{y_i}$ is the predicted value.
Lastly, an early stopping condition with patience $p=5$ is imposed to optimize variance and bias tradeoff. 


\subsection{Hyperparameter sweep}
\label{subsec:hyperparameter_sweep}

All other model hyperparameters (e.g., filter size and the dimensions of the dense layers, dropouts, etc) are determined using a \textit{hyperparameter sweep} which searches for the best possible hyperparameter values for each specific GVD and sample size.
The ranges of possible hyperparameter values are listed in Table \ref{tab:hyperpar}.
We initialized 50 unique CNN models by employing a randomized grid search over the hyperparameter space for each of the four separate GVD represented datasets.
This procedure is repeated each $N$ (the size of the data set) resulting in a total of 800 independently trained models visualized in Figure \ref{fig:hypersweep}.

\begin{table}[ht]
\centering
\renewcommand{\arraystretch}{1.2}
\scalebox{0.55}{
\begin{tabular}{c|c|c|c|c||c|c|c}
\toprule
 \multicolumn{5}{c||}{\textbf{Architecture Hyperparameters}} & \multicolumn{3}{c}{\textbf{Regularizers }}  \\
\cmidrule{1-4}\cmidrule{5-8}
 \textbf{Filter \#1} & \textbf{Filter \# 2} & \textbf{Layer \# 1} & \textbf{Layer \# 2} & \textbf{Layer \# 3} & \textbf{Dropout} & \textbf{$L_1$ } & \textbf{$L_2$ }\\
\midrule
   &    &     &     &     &  0.1  & 0.1   & 0.1  \\
8  & 8  & 128 & 64  & 32  &  0.2  & 0.2   & 0.2  \\
16 & 16 & 256 & 128 & 64  &  0.3  & 0.3   & 0.3  \\
32 & 32 & 512 & 256 & 128 &  0.5  & 0.5   & 0.5    \\
   &    &     & 512 & 256 &  0.75 & 0.75  & 0.75    \\
   &    &     &     &     &       & 1.0   & 1.0    \\
\bottomrule
\end{tabular}
}
\caption{List of possible hyperparameters used for random grid sweep.}
\label{tab:hyperpar}
\end{table}


\begin{figure*}
    \includegraphics[width=1\textwidth]{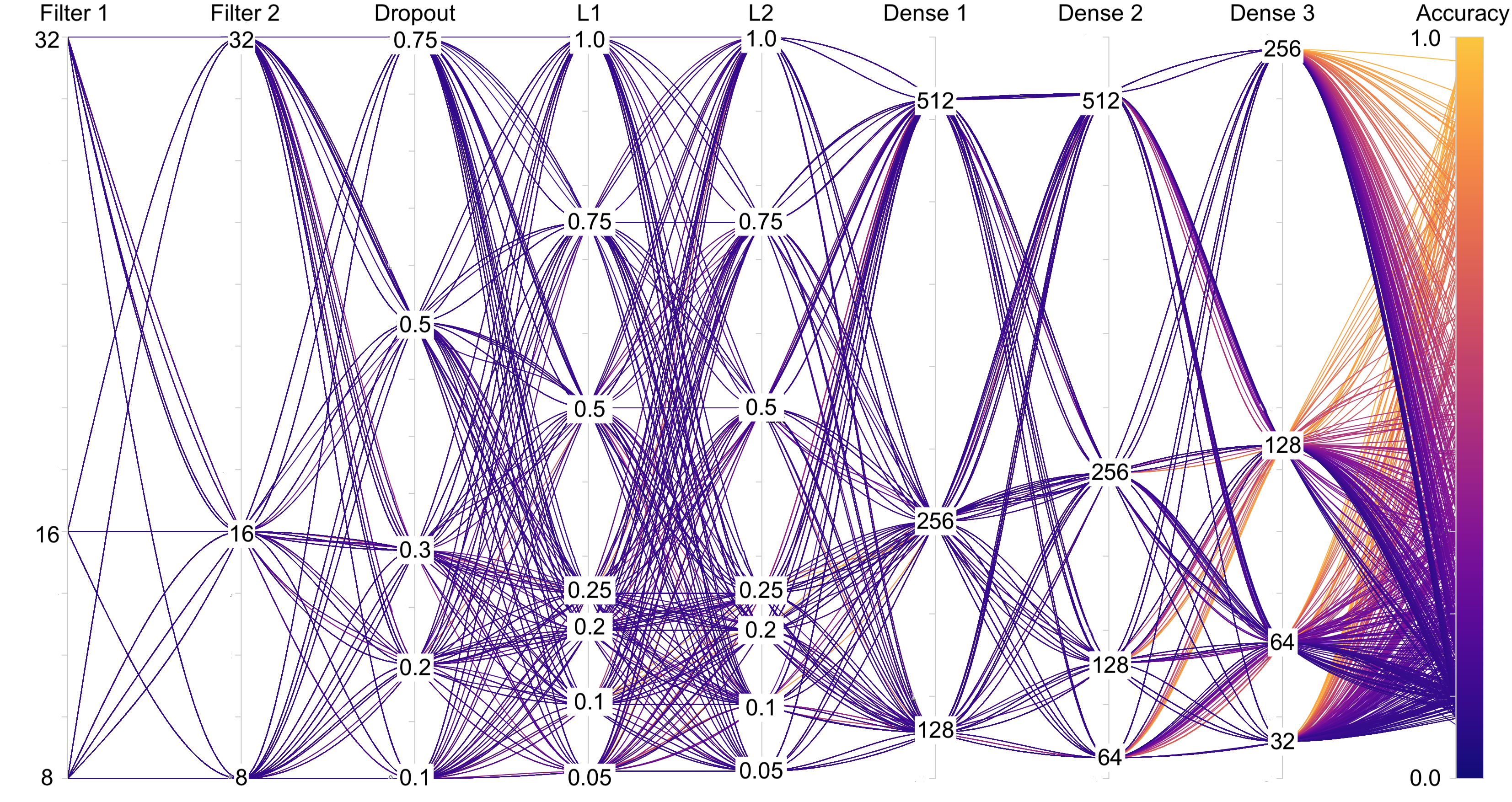}
    \setlength{\abovecaptionskip}{-5pt}
    \setlength{\belowcaptionskip}{-5pt}
    \caption{The figure above is a visualization of the hyperparameter sweep performed over the four voxel-associated data representations.
    Each curve represents a single model and its hyperparameters correspond to the intersection of the curve with the vertrical axes.
    Curves that are closer in shape to yellow correspond to models that yield high model performance, with respect to validation accuracy, whereas purplish curves are sub-optimal. Image created using Weights and Biases \cite{wandb}.}
    \label{fig:hypersweep}
\end{figure*}

\begin{table*}[h]
\centering
\renewcommand{\arraystretch}{1.2}
\scalebox{0.75}{
\begin{tabular}{c||c|c|c|c|c||c|c|c||c|c|c }
\toprule
 \multirow{2}{*}{\textbf{GVD}} &
 \multicolumn{5}{c||}{\textbf{Architecture Hyperparameters}} & \multicolumn{3}{c||}{\textbf{Regularizers }} & \multicolumn{3}{c}{\textbf{Model Performance }}  \\

\cmidrule{2-6}\cmidrule{7-10}\cmidrule{11-12}
 &\textbf{Filter \#1} & \textbf{Filter \# 2} & \textbf{Dense \# 1} & \textbf{Dense \# 2} & \textbf{Dense \# 3} & \textbf{Dropout} & \textbf{$L_1$ } & \textbf{$L_2$ } & \textbf{Train} & \textbf{Val}& \textbf{Test}\\
\midrule
\text{FID} 12 & 8 & 8  & 512 & 256  & 32  &  0.3 & 0.01 & 0.25 & 53.6\% & 49.3.9\% & 56.9\% \\
 \text{IPF} 12 & 32 & 8 & 256 & 64 & 64  &  0.75 & 0.01 & 0.5 & 56.5\% & 59.9\% & 55.4\% \\
 \text{QP} 12 & 8 & 16 & 512 & 64 & 256 &  0.1 & 0.25 & 0.5 & 54.6\% & 57.8 \% & 57.6\% \\
 \text{SEQ} 12 & 32 & 8 & 128 & 64 & 64 &  0.1 & 0.01 & 0.1  & 96.9\% & 96.6\% & \textbf{98.8\%} \\
 \text{FID} 25 & 32 & 16  & 128 & 64  & 128  &  0.2 & 0.01 & 0.1 & 54.6\% & 45.7\% & 55.9\% \\
 \text{IPF} 25 & 8 & 16 & 256 & 128 & 256  &  0.2 & 0.01 & 0.1 & 95.6\% & 65.2\% & 86.8\% \\
 \text{QP} 25 & 32 & 8 & 256 & 256 & 256 &  0.3 & 0.01 & 0.1 & 57.1\% & 57.7\% & 57.8\% \\
 \text{SEQ} 25 & 16 & 16  & 256  & 512 & 128 &  0.5 & 0.5 & 0.5  & 93.3\% & 94.8\% & \textbf{96.3\%} \\
 \text{FID} 51 & 32 & 16 & 128 & 256  & 128  &  0.2 & 0.01 & 0.1 & 48.8\% & 46.7\% & 50.1\% \\
 \text{IPF} 51 & 32 & 32 & 512 & 64 & 64  &  0.1 & 0.01 & 0.2 & 94.4\% & 71.0\% & 89.0\% \\
 \text{QP} 51 & 8 & 8 & 512 & 256 & 32 &  0.3 & 0.1 & 1.0 & 54.3\% & 43.7\% & 56.4\% \\
 \text{SEQ} 51 & 16  &  16  &  128 & 512 & 256 &  0.5 & 0.1 & 0.01  & 94.8\% & 95.9\% & \textbf{97.69\%} \\
 \text{FID} 102 & 8 & 16  & 128 & 256  & 128 &  0.3 & 0.01 & 0.25 & 50.7\% & 43.8\% & 49.89\% \\
 \text{IPF} 102 & 8 & 32 & 256 & 64 & 128  &  0.2 & 0.1 & 0.25 & 90.8\% & 90.5\% & 90.1\% \\
 \text{QP} 102 & 8 & 32 & 512 & 128 & 256 &  0.2 & 0.2 & 0.1 & 55.7\% & 55.7\% & 54.8\% \\
 \text{SEQ} 102 &  32 &  32  & 512  & 128 & 32 &  0.1 & 0.1 & 1.0  & 96.1\% & 97.8\% & \textbf{97.5\%} \\
\bottomrule
\end{tabular}
}
\caption{The table above highlights the architecture of the best performing model for each of the voxel-associated data and their corresponding sample size. We bold the testing accuracy for the four best performing models for each sample size. The results demonstrate that the SEQ is consistently the best performing data type.}
\label{tab:dataset_table}
\end{table*}


\subsection{Results}

\begin{figure}
    \centering
    \includegraphics[width=\columnwidth]{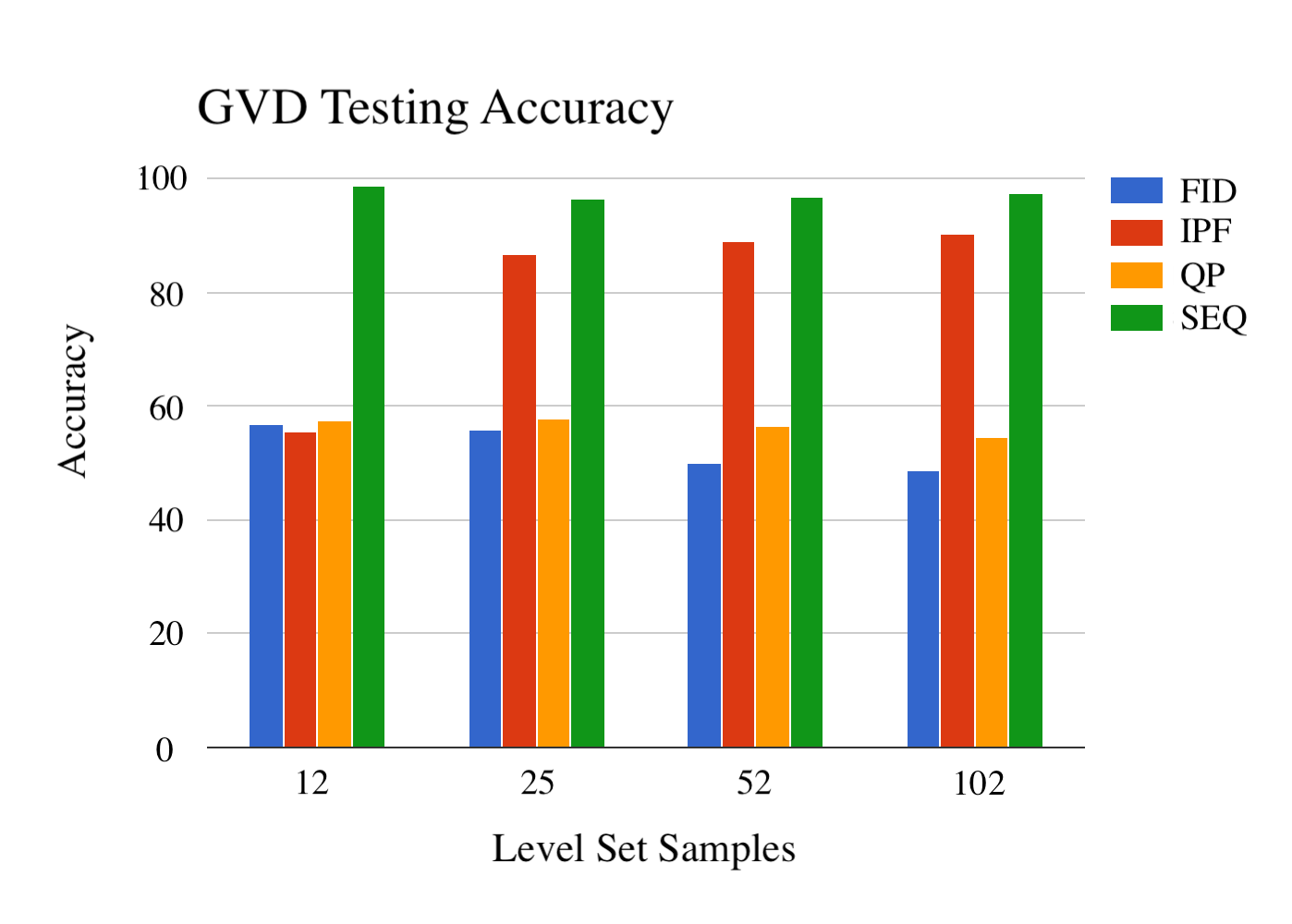}
    \caption{\label{fig:bargraph} Each bar in the graph corresponds to the best testing accuracy of each GVD representation with increasing sample size. In all cases the SEQ representation yields the best result regardless of sample size.}
\end{figure}

The results of our study are summarized in Table \ref{tab:dataset_table} and Figure \ref{fig:bargraph}, which displays the optimal hyperparameters for each GVD and data sample size, together with the training, validation, and test accuracies.
Our analysis indicated that the spectral embedding representation (SEQ) of the data consistently outperforms all other GVD representations in all metrics.
The second best performing GVD representation is the IPF; however, this performance is only accessible with sufficiently large sample size.
This suggests that IPF is not a viable alternative when data scarcity is an issue (that is, when working with physically generated image data rather than synthetic data).


\section{Discussion}
\label{sec:discussion}

It is readily apparent from Figure~\ref{fig:bargraph} that the FID \, parametrization of the microstructure is inadequate for the classification task. This can be ascribed to the fact that FID can only identify geometric information, and is blind to crystalline orientations.
A CNNs able to perfectly characterize the geometric features of these materials would accurately classify the materials with rolled and bimodal structures, but would randomly assign the three equiaxed materials to classes, resulting in an overall expected testing accuracy of 60\%. This is consistent with the results we obtained.
Regardless, it is clear that additional information is necessary to parse out the texture of the microstructure. 

In contrast to the FID, the CNNs trained on IPF representations improve substantially with increased volumes of data. 
The testing performance essentially saturates for a total data pool of 25 slices. 
It is possible that despite further data increases by 2, 4, 6, and 8- fold, the lower performance of IPF relative to SEQ arises due to the inherent limitaitons of the former.
While the IPF is widely adopted and convenient, it loses information by projecting the three-dimensional homogeneous space $SO(3)/H$ to two dimensional space. 
The CNN will inherently see IPF representations of these two planes as degenerate. In this way, IPF values are unable to fully capture the texture of a microstructure without additional information.

Given that quaternions faithfully represent the orientation of bodies in three-dimensional space, the poor performance of the CNNs trained on QP data may be somewhat surprising. This is likely due to the redundancies in the QP representation. A single crystalline orientation class corresponds to 48 distinct QP representations. By neglecting to mod out for rotational symmetry, we require the CNN to learn that each of these representations --- which are distant elements of the three-sphere $S^3$ --- are in fact the same. 


The performance of the SE far exceeds that of the other representations. The CNNs trained on SE data were able to classify all five structures with testing accuracies in the upper-90s for all quantities of data shown. While the IPF performs almost as well for the three larger data sets, the data set size requirements for the SE representation are an order of magnitude smaller. This dramatic increase in performance can be explained by the fact that the SE captures significantly more information by projecting $SO(3)/H$ into a four-dimensional space rather than a two-dimensional one. 

Though synthetic data was used for these exploratory studies, the availability of real microstructure data can be quite limited. 
Materials data acquisition beyond surface characterization is often destructive of the sample at hand, requiring fabrication of new samples for each scan.
The fabrication of new samples can be both monetarily and temporally expensive.
The act of acquiring materials data is also often subject to the same costs.

The ability to accurately classify materials in the data-limited regime is paramount.
Assuming that 90\% accuracy is adequate for a given purpose (which is often not the case), the acquisition of 100 EBSD slices necessitated for the IPF representation will be prohibitive for a general use case.
For these reasons, the SE is strictly preferable to the alternatives.

\label{sec:conclusion}


\section*{Acknowledgments}
S.P., T.B., and B.S.\ are grateful for the support of the National Science Foundation under Award No.\ 2232967.
D.M.\ and J.K.M.\ are grateful for the support of the National Science Foundation under Award No.\ 2232968.

\appendix

\section{The Spectral Embedding}

We explain how to find an orthonormal basis for $L^2\left(G/H\right)$ where $H$ is subgroup of a compact Lie group $G.$ First, we recall a  well-known theorem from representation theory.
\begin{theorem}[Peter–Weyl Theorem for Compact Groups] Every function $f \in L^2(G)$ has a “Fourier series,” converging in the $L^2$ sense, \[f= \sum_{\alpha \in \widehat{G}} \sum_{i,j=1}^{\text{dim}\,\D^{\alpha}}c_{i,j}^{\alpha}\D_{i,j}^{\alpha}\] where the $\D^{\alpha}$ are unitary representatives of the classes of inequivalent irreducible representations of $G$, the $\D^{\alpha}_{i,j}$ are their matrix coefficients in orthonormal bases, and \[c^{\alpha}_{i,j} = (\text{dim}\,\D^{\alpha})\, \langle f,\D_{i,j}^{\alpha} \rangle  = (\text{dim}\,\D^{\alpha})\, \int_{G}f(g)\overline{\D_{i,j}^{\alpha}(g)}dg \] In other words, the functions $\D_{i,j}^{\alpha}$ form an orthonormal basis for $L^2\left(G\right).$ 
\end{theorem}
Observe that if $f\in L^2\left(G\right)$ then the $H$-averaged function  
\[ \tilde{\D}^{\alpha}_{ij} = \frac{1}{\mu(H)} \int_{H} h\cdot \D^{\alpha}_{ij} \, d\mu  \] 
is $H$-invariant and thus descends to a well-defined function on $G/H.$ Here,  $\mu$ is the Haar measure on $H$ induced by the one on $G$ (when $H$ is finite, $\mu$ is the uniform measure on $H$). An orthonormal basis for $L^2\left(G/H\right)$ is constructed as follows. For each irreducible representation $\alpha\in\widehat{G},$ use the Gram--Schmidt process to find an orthonormal basis $f^{\alpha}_1,\ldots f^{\alpha}_{n\left(\alpha\right)}$ for the span of the $H$-averages of matrix coefficients of $\alpha,$ $\mathrm{span}\left(\tilde{\D}^{\alpha}_{i,j}\right).$ We claim that $\left\{f^{\alpha}_1\ldots, f^{\alpha}_{n\left(\alpha\right)}\right\}_{\alpha\in\widehat{G}}$ is an orthonormal basis for $L^2\left(G/H\right).$ These functions span $L^2\left(G/H\right)$ by the preceding theorem, for if $f\in L^2\left(G/H\right)$ then $f$ pulls back to an $H$-invariant $f_0\in L^2\left(G\right),$
$$f_0=\tilde{f_0}=\sum_{\alpha \in \widehat{G}} \sum_{i,j=1}^{\text{dim}\,\D^{\alpha}}c_{i,j}^{\alpha}\tilde{\D}_{i,j}^{\alpha}\,,$$
and each function $\tilde{\D}_{i,j}^{\alpha}$ can be expressed in terms of the functions $f_k^{\alpha}.$ By construction, $\left\{f^{\alpha}_1\ldots, f^{\alpha}_{n\left(\alpha\right)}\right\}$ is an orthonormal set. All that remains to be shown is that the functions obtained for different values of $\alpha$ are orthogonal. This follows from the next lemma and the Peter--Weyl Theorem.

\begin{lemma} Fix $\alpha\in \widehat{G}$ and let $f\in \mathrm{Span}\{\D_{ij}^{\alpha}\}$; that is, $f$ is a linear combination of the matrix coefficients corresponding to the representation $\alpha.$  Then for any $h\in G,$ the function $h\cdot f:G\to \R$ given by $\left(h\cdot f\right)\left(g\right)=f\left(h^{-1}g\right)$ is also contained in $\text{Span}\{\mathcal{D}_{ij}^{\alpha}\}\,.$  
\end{lemma}

\begin{proof} It is sufficient to show that if $f(g) = [\D^{\alpha}(g)]_{ik}$, then $f(h^{-1}g)$ is a linear combination of matrix entries of $f(g)$. Observe that \[\begin{aligned}
    [\D^{\alpha}(h^{-1}g)]_{ik} &=  [\D^{\alpha}(h^{-1})\D^{\alpha}(g)]_{ik} \\
    &= \left(i^{\text{th}}\text{ row of }\D^{\alpha}(h^{-1})\right)\cdot \left(k^{\text{th}}\text{ column of }\D^{\alpha}(g)\right) \\
    &= \sum_{j}[\D^{\alpha}(h^{-1})]_{ij}[\D^{\alpha}(g)]_{jk} \qedhere
\end{aligned}\]
\end{proof}

\bibliographystyle{elsarticle-num}
\bibliography{refs}

\end{document}